\documentclass{amsart}

\allowdisplaybreaks
\numberwithin{equation}{section}

\theoremstyle{plain} \newtheorem{theorem}{Theorem}[section]
\theoremstyle{plain} 
\theoremstyle{plain} 
\theoremstyle{plain} \newtheorem{proposition}[theorem]{Proposition}
\theoremstyle{remark} 
\theoremstyle{definition} 
\theoremstyle{definition} 
\theoremstyle{remark}

\newcommand{\alp}{\alpha}

\newcommand{\be}{\beta}

\renewcommand{\Pr}{ \mathbb P}

\newcommand{\lk}{\lambda_k}
\newcommand{\lkm}{\lambda_{k-1}}
\newcommand{\lz}{\lambda_0}

\newcommand{\lu}{\lambda^{\uparrow}}
\newcommand{\lr}{\lambda^{\rightarrow}}

\newcommand{\bP}{\mathbb{P}}

\newcommand{\bR}{\mathbb{R}}

\newcommand{\cP}{\mathcal{P}}

\begin{document}

\title[Correlated default]{Non-existence of Markovian time dynamics \\
for graphical models of correlated default}

\author{Steven N. Evans} 

\address{Department of Statistics\\
         University  of California\\ 
         367 Evans Hall \#3860\\
         Berkeley, CA 94720-3860 \\
         U.S.A.}

\email{evans@stat.berkeley.edu}
\thanks{SNE supported in part by NSF grants DMS-0405778 and DMS-0907630}

\author{Alexandru Hening} 

\address{Department of Mathematics\\
         University  of California\\
         970 Evans Hall \#3840\\
         Berkeley, CA 94720-3840\\
         U.S.A.}

\email{ahening@math.berkeley.edu}



\date{\today}

\begin{abstract}
Filiz et al. (2008) proposed a model for the pattern of defaults seen among a group of firms
at the end of a given time period.  The ingredients in the model are a graph $G = (V,E)$, where
the vertices $V$ correspond to the firms and the edges $E$ describe
the network of interdependencies between the firms, a parameter for each vertex that 
captures the individual propensity of that firm to default, 
and a parameter for each edge that captures the joint propensity of the
two connected firms to default.  
The correlated default model can be re-rewritten as a standard
 Ising model on the graph by identifying the set of
defaulting firms in the default model with the set of sites
in the Ising model for which the $\{\pm 1\}$-valued spin is $+1$.  We ask whether there is
a suitable continuous time Markov chain $(X_t)_{t \ge 0}$ taking values in the subsets of $V$
such that $X_0 = \emptyset$, $X_r \subseteq X_s$ for $r \le s$ (that is, once a firm defaults it stays
in default), the distribution of $X_T$ for some fixed time $T$ is the
one given by the default model, and the distribution of $X_t$ for other times $t$ is described
by a probability distribution in the same family as the default model.  In terms of the equivalent Ising model,
this corresponds to asking if it is possible
to begin at time $0$ with a configuration in which every spin is $-1$
and then flip spins one at a time from $-1$ to $+1$ according to Markovian dynamics
so that the configuration of spins at each time is described by some Ising model
and at time $T$ the configuration is distributed according to the prescribed Ising model.  
We show for three simple but
financially natural special cases that this is not possible outside of the trivial case where
 there is complete independence between the firms.
\end{abstract}

\maketitle

\section{Introduction}

Investors are exposed to {\em credit risk} due to the possibility that
one or more counterparties in a financial agreement will {\em default};
that is, not honor their obligations to make certain payments.  
Some examples of default
are a consumer or business not making a due payment on a loan,
a manufacturer or retailer not paying  for goods already received
from a supplier, a bond issuer not making coupon or principal payments,
or an insolvent financial institution not returning deposited
funds to its customers upon demand.

Some credit risk is present in virtually 
any financial agreement,
and a key ingredient in its satisfactory
management is a model that produces a sufficiently
accurate probability for a given default event. 
Consequently, there is a large theoretical and applied literature
on this topic \cite{BR02, BOW02, DS03, dSO04, Gie04, Sch04, ZP07, Wag08}.  
Roughly speaking, models of default lie on a spectrum
between the {\em structural} and {\em reduced form} ones.  For the
example of a firm defaulting on its debt obligations, a structural
model might include explicit descriptions of the dynamics of the firm's assets, capital holdings and debt structure,
whereas a reduced form model would not seek to incorporate
the details of the actual mechanism by which the firm is led to
default but rather it might typically be something
of a ``black box'' that treats the time of default as a random
time with an associated exogenous intensity process having a
rather simple structure characterized by a small number of parameters which may have little direct economic interpretation.
Although structural models are perhaps theoretically
more satisfying because in principle they provide a means of testing how
well the factors that cause default are understood, they
are often perceived as being too complex and parameter-rich
for them to be fitted adequately: defaults
are uncommon and even firms within the same
sector of the economy can be quite heterogeneous,
so there can be insufficient ``independent replication''
upon which to base statistically sound parameter estimates.

The difficulty of modeling default probabilities is compounded for
complex financial instruments such as {\em collateralized debt obligations} (CDOs)
and other structured asset-backed securities that are constructed
by, in essence, bundling together a group of borrowers. It is
then no longer sufficient to determine the default probability
for a single ``firm'' --  rather, it becomes necessary to model the joint
probabilities that various subsets of firms in the
basket will default, and it is
usually not appropriate to treat the defaults of different
firms as statistically independent events.  The most
obvious reason for this absence of independence is that
firms are subject to the same background economic environment.
Moreover, situations such as the interconnectedness between
manufacturers, their parts suppliers and the retailers who
sell their products can cause problems for one firm in such a network
to spread to others via a process that 
is usually described as {\em contagion}.  A small sampling
of the substantial empirical and modeling work on this phenomenon
of {\em correlated default} is 
\cite{DV01a, DV01b, HW01, SS01, Zho01, FMcN03, Gie03,
LG03, Gie04b, GW04, GW06, DDKS07, EGG07, JZ07, Yu07, ES09, GGD09, CDGH10}.
 
In this paper we investigate a particularly appealing class of models for 
correlated default in \cite{FGMS08} (see also
\cite{MV05, KMH06}, where special cases of this
model were introduced).  The basic model in \cite{FGMS08}
does not attempt to describe the time course of defaults for
some group of firms.  Rather, it is a {\em one period} model
that gives the probability any given subset of the firms will have
defaulted at some time during a prescribed time interval.

The ingredients of the model in \cite{FGMS08} are a finite
(undirected, simple) graph $G$ with vertex set $V$ and edge set $E$
and two vectors of parameters $\alpha = (\alpha_v)_{v \in V} \in \bR^V$
and $\beta = (\beta_e)_{e \in E} \in \bR^E$.  Each vertex $v \in V$
represents a firm and the graph structure provided by the edges
is intended to capture the network of interdependencies between
the firms.  Write $I_v$, $v \in V$, for the indicator random
variable of the event that firm $v$ defaults; that is,
$I_v$ takes the value $1$ if firm $v$ defaults and the value
$0$ otherwise.  The probability of a given pattern
$\varepsilon = (\varepsilon_v)_{v \in V} \in \{0,1\}^V$ of defaults is
\begin{equation}
\label{E:model_def}
\bP\{I_v = \varepsilon_v, \; v \in V\}
:=
\frac{1}{Z} \exp(H(\varepsilon)),
\end{equation}
where the {\em Hamiltonian} $H$ is given by
\begin{equation}
\label{E:Hamiltonian_def}
H(\varepsilon) := 
\sum_{u \in V} \alpha_u \varepsilon_u
+
\sum_{\{v,w\} \in E} \beta_{\{v,w\}} \varepsilon_v \varepsilon_w
\end{equation}
and the {\em partition function} $Z$ is  the
normalizing constant that ensures the sum over $\{0,1\}^V$ 
of the probabilities is one.  The parameter $\alpha_u$, $u \in V$,
is clearly some measure of the individual propensity of
firm $u$ to default. The parameter $\beta_{\{v,w\}}$,
$\{v,w\} \in E$, captures in some way the dependence between
the defaults of firm $v$ and firm $w$: if this parameter
is positive, then the joint default of both firms is favored,
whereas it is discouraged when the parameter is negative.
We write $\cP(G,\alpha,\beta)$ for the distribution of 
the random binary vector $I$.

Note that if we set $Y_v = 2 I_v - 1$, $v \in V$, then
$(Y_v)_{v \in V} \in \{\pm 1\}^V$ and for 
$(\sigma_v)_{v \in V} \in \{\pm 1\}^V$ we have
\[
\begin{split}
& \bP\{Y_v = \sigma_v, \; v \in V\}
 =
\bP\{I_v = (1+\sigma_v)/2, \; v \in V\} \\
& \quad =
\frac{1}{Z}
\exp
\left(
\frac{1}{2} \sum_{u \in V} \alpha_u (1+\sigma_u)
+
\frac{1}{4} \sum_{\{v,w\} \in E} \beta_{\{v,w\}} (1+\sigma_v)(1+\sigma_w)
\right) \\
& \quad =
\frac{1}{\tilde Z}
\exp
\left(
\sum_{u \in V} \gamma_u \sigma_u
+
\sum_{\{v,w\} \in E} \delta_{\{v,w\}} \sigma_v \sigma_w
\right) \\
\end{split}
\]
for suitable parameters 
$(\gamma_v)_{v \in V} \in \bR^V$
and $(\delta_e)_{e \in E} \in \bR^E$ and a corresponding normalization
constant $\tilde Z$.  Thus, the random
vector of {\em spins} $(Y_v)_{v \in V}$ is described by the usual
{\em Ising model} associated with the graph $G = (V,E)$.

It is shown in \cite{FGMS08} that this class of correlated default models
is as flexible as one could possibly hope:
if $J$ is an arbitrary
$\{0,1\}^V$-valued random variable, then there is a choice
of the parameters $(\alpha_v)_{v \in V}$
and $(\beta_e)_{e \in E}$ such that 
$I_u$ has the same distribution as $J_u$ for all $u \in V$
and for all $\{v,w\} \in E$ 
the pair $(I_v, I_w)$ has the same
distribution as $(J_v, J_w)$.  Moreover, it is observed
in \cite{FGMS08} that it is possible to fit such a model to
data using existing techniques such as iterative proportional
fitting, various convex optimization techniques, or a number of
other ``off-the-shelf'' numerical optimization methods suitable
for large-scale computation.

A significant drawback of the class of models in \cite{FGMS08}
is that they don't provide a description of the time dynamics
of default: they just give the probability that a given subset
of firms have defaulted during some fixed time period without
saying anything about the distribution of the times at which
the defaults occurred.  If we let $[0,T]$ be the time period of interest,
then we would like there to be a $\{0,1\}^V$-valued stochastic process 
$(I(t))_{0 \le t \le T}$ such 
that
\begin{itemize}
\item
$I_v(t)=1$ if and only if firm $v \in V$ has defaulted by time $t$,
so that $I_v(0) = 0$ and the sample paths of $(I_v(t))_{0 \le t \le T}$ 
are right-continuous and non-decreasing (once a firm defaults it does not
``undefault''),
\item
$\#\{v \in V : I_v(t) \ne I_v(t-)\} \le 1$ for any $t \in [0,T]$ (two or more
firms do not default simultaneously), 
\item
$I(T)$ is has distribution $\cP(G,\alpha,\beta)$.
\end{itemize}
Furthermore, since $\cP(G,\alpha,\beta)$ is supposed to be an appropriate
description for the pattern of defaults during $[0,T]$, it is reasonable
to require that
\begin{itemize}
\item
$I(t)$ has distribution $\cP(G,\alpha(t),\beta(t))$ for suitable
parameters $\alpha(t)$ and $\beta(t)$ when $0 < t < T$.
\end{itemize}

In this paper we investigate whether such a process exists within the simplest
and perhaps most natural class of models, namely the time-homogeneous Markov chains.
Re-cast in the language of the equivalent Ising model, 
we are thus asking if it is possible
to begin at time $0$ with a configuration in which every spin is $-1$
and then flip spins one at a time from $-1$ to $+1$ according to Markovian dynamics
so that the configuration of spins at time $T$ is distributed according to
a prescribed Ising model and at all other times the configuration is described by some Ising model.

We can certainly construct such a chain if $\beta = 0$, so that
$\cP(G,\alpha,\beta) = \cP(G,\alpha,0)$ is the distribution of a vector $(I_v)_{v \in V}$
of independent $\{0,1\}$-valued random variables with
\[
\bP\{I_v = 0\} = \frac{1}{1 + \exp(\alpha_v)}.
\]
We simply takes the processes $(I_v(t))_{t \ge 0}$ to be independent,
with
\[
\bP\{I_v(t) = 0\} = \exp(-\lambda_v t),
\]
where the jump rate $\lambda_v$ is chosen so that
\[
\exp(-\lambda_v T) = \frac{1}{1 + \exp(\alpha_v)}.
\]
Thus, $\lambda_v = \frac{1}{T} \log(1 + \exp(\alpha_v))$ and
$I(t)$ has distribution $\cP(G, \alpha(t), 0)$, where
\[
\frac{1}{1 + \exp(\alpha_v(t))} 
= \exp(-\lambda_v t) 
= \exp\left(-\frac{t}{T} \log(1 + \exp(\alpha_v))\right),
\]
so that
\[
\alpha_v(t) = \log\left((1 + \exp(\alpha_v))^{\frac{t}{T}} - 1\right)
\]
for $0 < t < T$.

After establishing some general facts in Section~\ref{S:generalities}, we investigate
in Sections~\ref{S:Model_I}, \ref{S:Model_II} and \ref{S:Model_III}
whether it is possible to construct a time-homogeneous
Markov chain for non-zero $\beta$ in the 
following cases:
\begin{itemize}
\item[(I)]
$G$ is the complete graph $K_N$ in which there are $N$ vertices with each
vertex connected to every other one, $\alpha_{u} = \alpha_{v}$ for $u,v \in V$
and $\beta_{e} = \beta_{f}$ for $e, f \in E$;
\item[(II)]
$G$ is the complete bipartite graph $K_{M,N}$ in which $V$ is partitioned
into two disjoint subsets $\hat V$ and $\check V$ of cardinality $M$ and $N$ such
that every vertex in $\hat V$ is connected to every vertex in $\check V$ and there
are no other edges, $\alpha_{u} = \alpha_{v}$ for $u,v \in V$
and $\beta_{e} = \beta_{f}$ for $e, f \in E$;
\item[(III)]
$G$ is again the complete bipartite graph $K_{M,N}$, 
$\alpha_{u} = \alpha_{v}$ for $u,v \in \hat V$,
$\alpha_{u} = \alpha_{v}$ for $u,v \in \check V$
and $\beta_{e} = \beta_{f}$ for $e, f \in E$.
\end{itemize}
In Model I there is complete symmetry: each firm has the same individual
propensity to default and the interdependence between any two firms
is the same as that between any other two.  Model II and III both
describe a situation in which there are two types of firms (say,
for example, car manufacturers and auto parts suppliers) and there
is only interdependence between firms of different types.  In Model II
all firms have the same individual propensity to default, whereas
in Model III this propensity can depend on the type of the firm.

We conclude in all three cases (with a minor technical restriction
for Model III) that it is impossible to construct
a time-homogeneous Markov chain with the desired properties unless
$\beta$ is zero; that is, unless the firms behave independently.

\section{Generalities}
\label{S:generalities}

It will be notationally more convenient to identify
a vector $\varepsilon = (\varepsilon_v)_{v \in V} \in \{0,1\}^V$
with the subset $A = \{v \in V : \varepsilon_v = 1\} \subseteq V$
and regard $\cP(G, \alpha, \beta)$ as a probability measure
on subsets of $V$ rather than $\{0,1\}^V$.  If we extend the definition
of $\beta_{\{u,v\}}$ by declaring that $\beta_{\{u,v\}} = 0$ when
$\{u,v\} \notin E$ and write $\beta_{\{u,v\}}$ more simply
as $\beta_{uv}$, then our Hamiltonian, now thought of as function
defined on subsets of $V$, is given by
\begin{equation}
\label{e1}
H(A):= \sum_{u\in A} \alpha_u + \sum_{\{u,v\} \subset A} \beta_{uv}.
\end{equation}
If we write $\bP^H$ for the probability measure $\cP(G, \alpha, \beta)$, then
\begin{equation}\label{e2}
\Pr^H(\{A\})
:=
\frac{1}{Z} \exp(H(A)),
\end{equation}
where
\[
Z 
:= 
\sum_{B \subseteq V} \exp(H(B)).
\]

We are interested in the existence of a time-homogeneous Markov chain
$X = (X_t)_{t \ge 0}$ that has as its state-space the collection of
subsets of $V$ and has the following properties, where we write
$Q(A,B)$ for the jump rate from state $A$ to state $B$:
\begin{itemize}
\item 
$Q(A,B) = 0$ unless $B = A \cup \{v\}$ for some $v \notin A$;
\item 
when $X(0) = \emptyset$, the distribution of $X(T)$ is $\Pr^H$;
\item
there are parameter vectors $\alpha(t)$ and $\beta(t)$ for $0 < t \le T$ such that
if we set
\[
H_t(A):= \sum_{u\in A} \alpha_u(t) + \sum_{\{u,v\} \subset A} \beta_{uv}(t),
\]
then $X(t)$ has distribution
$\Pr^{H_t}$ when $X(0)=\emptyset$.
\end{itemize}
If such a Markov chain exists, we say that the default model {\em admits
time-homogeneous Markovian dynamics}.

Write $A \rightarrow B$ if $B = A \cup \{v\}$ for some $v \notin A$.
The Kolmogorov forward equations for the chain $X$ with initial
state $\emptyset$ become
\begin{equation}\label{e3}
\frac{d}{dt} \Pr^{H_t}(B)
=
\sum_{A\rightarrow B} \Pr^{H_t}(A) Q(A,B) + \Pr^{H_t}(B) Q(B,B),
\end{equation}
where, as usual, we put $Q(B,B) := - \sum_{C \ne B} Q(B,C)$.

Denoting the partition function associated with the
Hamiltonian $H_t$ by  $Z_t := \sum_{C\in E} e^{H_t(C)}$, we have
\[
\begin{split}
\frac{d}{dt} \Pr^{H_t}(B)
& = 
\frac{d}{dt} \frac{e^{H_t(B)}}{Z_t} \\
& =
\frac{
Z_t e^{H_t(B)} H_t'(B) 
- Z_t' e^{H_t(B)}
}
{
Z_t^2
},\\
\end{split}
\]
and thus \eqref{e3} can be re-written as
\begin{equation}
\label{e4}
Z_t e^{H_t(B)} H_t'(B) 
- Z_t' e^{H_t(B)}
=
\sum_{A\rightarrow B} Q(A,B) e^{H_t(A)} Z_t 
+ 
Q(B,B) e^{H_t(B)} Z_t.
\end{equation}

To further simplify notation, set $R_B = - Q(B,B)$. 
Because $H_t(\emptyset) = 0$,
equation \eqref{e4} for $B=\emptyset$ is simply
\begin{equation*}
- Z_t' = - R_{\emptyset} Z_t.
\end{equation*}
We require 
\[
1 
= \lim_{t \downarrow 0} \Pr^{H_t}(\emptyset)
= \lim_{t \downarrow 0} \frac{1}{Z_t},
\]
and so
\begin{equation}
\label{e5}
Z_t = e^{R_\emptyset t}.
\end{equation}
Substituting \eqref{e5} into \eqref{e4} gives
\begin{equation}
H_t'(B) 
=
\sum_{A\rightarrow B} Q(A,B) e^{H_t(A) - H_t(B)} 
+ R_\emptyset 
- R_B,
\end{equation}
and hence
\begin{equation}
\label{e6}
\begin{split}
& \sum_{u\in B} \alpha_u'(t) \, + \sum_{\{u,v\} \subseteq B} \beta_{uv}'(t) \\
& \quad = 
\sum_{u \in B} Q(B \setminus \{u\},B ) 
\exp\left(-\alpha_u(t) - \sum_{v\in B \setminus \{u\}}\beta_{uv}(t)\right) \, + \, R_\emptyset \, - \, R_B. \\
\end{split}
\end{equation}


For $u \in V$, set $Q_u = Q(\emptyset, \{u\})$ and $R_u := R_{\{u\}} = - Q(\{u\}, \{u\})$.
Equation \eqref{e6} for $B=\{u\}$ is
\begin{equation}
\label{e8}
\alpha_u '(t) = Q_u e^{-\alpha_u (t)} + R_\emptyset - R_u.
\end{equation}
Hence, by the method of variation of parameters (also called
variation of constants),
\begin{equation}
\label{e10}
\alpha_u(t) = \log\left(\frac{Q_u}{R_\emptyset-R_u}\left(e^{(R_\emptyset-R_u)t}-1\right)\right)
\end{equation}
and
\begin{equation}
\label{e10.5}
\alpha_u'(t) 
= 
\frac{R_\emptyset - R_u}
{1 - e^{-(R_\emptyset - R_u)t}}
\end{equation}
when $R_\emptyset \ne R_u$.  If $R_\emptyset = R_u$, then
\begin{equation}
\label{e10'}
\alpha_u(t) = \log\left(Q_u t\right)
\end{equation}
and
\begin{equation}
\label{e10.5'}
\alpha_u'(t) 
= 
\frac{1}{t}.
\end{equation}
Note that each function $\alpha_u$, $u \in V$, is completely determined by the rates 
$Q_u = Q(\emptyset, \{u\})$, $R_\emptyset = \sum_{v \in V} Q(\emptyset,\{v\})$,
and $R_u = \sum_{v \in V \backslash \{u\}} Q(\{u\}, \{u,v\})$, and hence
the vector of functions $(\alpha_u)_{u \in V}$ is completely determined by the
collection of rates 
$\{Q(\emptyset, \{u\}): u \in V\} \cup \{Q(\{u\}, \{u,v\}): u,v \in V, \, u \ne v\}$.

For $u,v \in V$, set 
$Q_{uv}:=Q({\{u\},\{u,v\}})$ and $R_{uv} := R_{\{u,v\}} = - Q(\{u,v\}, \{u,v\})$.
Equation \eqref{e6} for $B=\{u,v\}$ is, upon substituting from \eqref{e10.5},
\begin{equation}
\label{e9}
\begin{split}
\beta_{uv}'(t) 
& = Q_{vu} e^{-\alpha_u (t) - \beta_{uv} (t)} 
+   Q_{uv} e^{-\alpha_v (t) - \beta_{vu} (t)} \\
& \quad - \alpha_u'(t) - \alpha_v'(t) - R_{uv} + R_\emptyset \\
& = \frac{Q_{vu}}{Q_u} \frac{R_\emptyset - R_u}{1 - e^{(R_\emptyset - R_u)t}} e^{- \beta_{uv} (t)} 
+   \frac{Q_{uv}}{Q_v} \frac{R_\emptyset - R_v}{1 - e^{(R_\emptyset - R_v)t}} e^{- \beta_{vu} (t)} \\
& \quad - \frac{R_\emptyset - R_u}{1 - e^{-(R_\emptyset - R_u)t}} 
        - \frac{R_\emptyset - R_v}{1 - e^{-(R_\emptyset - R_v)t}}  \\
& \quad  + R_\emptyset -  R_{uv} \\
\end{split}
\end{equation}
when $R \ne R_u$ and $R \ne R_v$.  Analogous results hold when $R = R_u$ or $R = R_v$.
Recall that $\beta_{uv}(t) = \beta_{\{u,v\}}(t) = \beta_{vu}(t)$, and so \eqref{e9} is 
an ordinary differential equation for the function $\beta_{uv}$ if we treat the rates
of the Markov chain as given.  In particular, the two vectors of functions 
$(\alpha_u)_{u \in V}$ and
$(\beta_{uv})_{u,v \in V, u \ne v}$ are completely determined 
by the collection of rates 
$\{Q(\emptyset, \{u\}): u \in V\} 
\cup \{Q(\{u\}, \{u,v\}): u,v \in V, \, u \ne v\}
\cup \{Q(\{u,v\}, \{u,v,w\}): u,v,w \in V, \, u \ne v \ne w \ne u\}$.

In principle, we could attempt to find values for these rates such that 
$(\alpha_u(T))_{u \in V}$ and
$(\beta_{uv}(T))_{u,v \in V, u \ne v}$
have the required value, substitute the resulting values of $\alpha_u(t)$
and $\beta_{uv}(t)$ into \eqref{e6} (using \eqref{e10.5} or \eqref{e10.5'} 
for the values of $\alpha_u'(t)$
and \eqref{e9}
or its analogues when $R = R_u$ or $R = R_v$  for the values of $\beta_{uv}'(t)$)
and hope to either find values for the remaining rates so that \eqref{e6} holds
for all $B \subseteq V$ or show that this is impossible
no matter what our initial choice of rates was.  This seems to be a rather forbidding task
in general, but we are able to carry it out in the three special cases described in the Introduction.


\section{Model I: complete symmetry}
\label{S:Model_I}

Recall Model I from the Introduction.  The graph $G$ is $K_N$, the complete graph on $N$
vertices for some $N$, and there are functions $\alpha$ and $\beta$ such that
\begin{equation} \label{e12}
\left\lbrace
 	\begin{aligned}
		\alpha_u (t) &= \alpha (t) ~\text{for all}~ u \in V\\
		\beta_{uv} (t) &= \beta (t)~\text{for all}~ u,v \in V \, u \neq v.
	\end{aligned}
 \right.
 \end{equation}
 
\begin{proposition} 
\label{P:Model_I}
Model I with $N \ge 4$ admits time-homogeneous
Markovian dynamics if and only if the firms default independently.
\end{proposition}
 
\begin{proof}
We observed in the Introduction that the general default model
admits Markovian dynamics when firms default independently.
So we need to establish a converse for the special case of Model I
with $N \ge 4$.

Suppose that a collection of rates exists such that
\eqref{e6} holds for all subsets $B$.  When $\# B \ge 1$,
\eqref{e6} becomes
\begin{equation}
\label{e6_symmetric}
\begin{split}
\# B \alpha'(t) + \binom {\# B}{2} \beta'(t)
& =
\left[\sum_{u \in B} Q(B \setminus \{u\},B ) \right]
\exp\left(-\alpha(t) - (\#B - 1) \beta(t)\right) \\
& \quad  + \, R_\emptyset \, - \, R_B. \\
\end{split}
\end{equation}
If we average \eqref{e6_symmetric} over all $\binom{N}{k}$ choices of sets $B$
with $\#B = k$ for some $k \ge 1$ and set
\[
\lambda_\ell 
= \binom{N}{\ell}^{-1} \sum_{A \subseteq V, \# A = \ell} R_A, 
\quad 0 \le \ell \le N,
\]
we get the equations
\begin{equation}
\label{e15}
\begin{split}
k\alp'(t) + \binom{k}{2} \be'(t) 
& = (\lz-\lk)\\
& \quad + \lkm \frac{k}{N-k+1}e^{-\alp(t)-(k-1)\be(t)},
\quad 1 \le k \le N. \\
\end{split}
\end{equation}
Note that $\lambda_\ell > 0$ for $0 \le \ell \le N-1$ and $\lambda_N  = 0$.

Equation \eqref{e15} for $k=1$ and $k=2$ yields
\begin{equation}
\label{e16}
\left\lbrace
 	\begin{aligned}
		\alp'(t)&=(\lz-\lambda_1)+\frac{\lambda_0}{N}e^{-\alp(t)}\\
		\be'(t) &= (\lz-\lambda_2) + \lambda_1\frac{2}{N-1} e^{-\alp(t)-\be(t)} - 2(\lz-\lambda_1) - 2\frac{\lz}{N} e^{-\alp(t)}.
	\end{aligned}
 \right.
\end{equation}
Substituting the values for $\alp'(t)$ and $\be'(t)$ from \eqref{e16}
into \eqref{e15} gives a system of equations of the form
\begin{equation}\label{e17}
\lkm\frac{k}{N-k+1} e^{-(k-1)\be(t)} -\lambda_1\frac{2}{N-1}e^{-\be(t)}
= a_k e^{\alp(t)} + b_k, \quad 1 \le k \le N,
\end{equation}
for appropriate constants $a_k$ and $b_k$, $1 \le k \le N$,
that depend on the constants $\lambda_\ell$, $0 \le \ell \le N$.

We claim that the continuous function $\be$ is constant. 
Suppose that this is not so.
Note that $a_k$ can be non-zero for at most one value of $k \in \{1, \ldots, N\}$, because if $a_{k'}\neq 0$ and $a_{k''}\neq 0$ for $1 \le k' <  k'' \le N$, then
\begin{equation*}
\begin{split}
& \frac{\lambda_{k'-1}\frac{k'}{N-k'+1} e^{-(k'-1) \be(t)}-\lambda_1\frac{2}{N-1}x-b_{k'}}{a_k'} \\
& \quad = 
\frac{\lambda_{k''-1}\frac{k''}{N-k''+1} e^{-(k''-1) \be(t)}-\lambda_1\frac{2}{N-1}x-b_{k''}}{a_k''}, \\
\end{split}
\end{equation*}
and letting $t$ vary over an open interval $J$ such that the image $\{\be(t) : t \in J\}$
contains an open interval we would conclude that two polynomials of different
degrees coincided over an open interval.  Because $N \ge 4$, we thus must have $a_k=0$ for some
$k \ge 3$. Observe for such a $k$ that
\begin{equation*}
\lkm\frac{k}{N-k+1} e^{-(k-1)\be(t)} -\lambda_1\frac{2}{N-1}e^{-\be(t)}= b_k,
\end{equation*}
and again we would conclude that two polynomials of different
degrees coincided over an open interval.  Therefore, the function $\be$ must be a constant, say $\beta^*$. 
Of course, $\beta^*$ is the pre-specified value for $\beta(T)$.

We now show that $\beta^*=0$.
Equation \eqref{e15} now becomes
\begin{equation*}
k\alp'(t)  = (\lz-\lk)+\lkm \frac{k}{N-k+1}e^{-\alp(t)-(k-1)\beta^*}
\end{equation*}
and hence, by \eqref{e16},
\begin{equation}
\label{E:affine}
k \left[(\lz-\lambda_1)+\frac{\lambda_0}{N}e^{-\alp(t)} \right] 
= (\lz-\lk)+\lkm \frac{k}{N-k+1}e^{-\alp(t)-(k-1)\beta^*}.
\end{equation}
Each side of \eqref{E:affine} is first degree polynomial in $e^{-\alpha(t)}$
for every $k \in \{1, \ldots, N\}$.
It is apparent from the differential equation in \eqref{e16} that the function $\alpha$ is not constant 
(indeed, we solved this equation explicitly in \eqref{e10}).  Consequently, the coefficients
of these two polynomial coincide and hence
\begin{equation}
\label{E:coeff_conclusion}
\left\lbrace
 	\begin{aligned}
		k(\lambda_0-\lambda_1)&=(\lz-\lambda_k)\\
		k\frac{\lz}{N} &= \lambda_{k-1}\frac{k}{N-k+1}e^{-(k-1)\beta^*}
	\end{aligned}
 \right.
\end{equation}
for $1 \le k \le N$.
Re-arranging \eqref{E:coeff_conclusion}, we conclude that
\begin{equation*}
		\lambda_k =k(\lambda_1-\lz)+\lz
\end{equation*}
for $1 \le k \le N$ and
\begin{equation*}
\lambda_k = \frac{(N-k)}{N}\lz e^{k \beta^*}
\end{equation*}
for $0 \le k \le N-1$. Because $N \ge 4$, this
is impossible unless $\beta^* = 0$.
\end{proof}

\section{Model II: Two classes with common individual propensity to default}
\label{S:Model_II}


Recall Model II from the Introduction.  The graph $G$ is $K_{M,N}$, the complete bipartite 
graph with vertex set the disjoint union $V = \hat V \sqcup \check V$, where $\hat V$ has $M$ vertices, $\check V$ has $N$
vertices, and there are functions $\alpha$ and $\beta$ such that
\begin{equation} \label{e12_II}
\left\lbrace
 	\begin{aligned}
		\alpha_u (t) &= \alpha (t) ~\text{for all}~ u \in V\\
		\beta_{uv} (t) &= \beta (t)~\text{for all}~ u \in \hat V,  \, v \in \check V.
	\end{aligned}
 \right.
 \end{equation}

\begin{proposition}
\label{p_aa}
Model II with $M \ge 3$ or $N \ge 3$ admits time-homogeneous
Markovian dynamics if and only if the firms default independently.
\end{proposition}

\begin{proof}  
As in the proof of Proposition~\ref{P:Model_I}, it suffices from the remarks
made in the Introduction about the general model
to show that if the model admits time-homogeneous
Markovian dynamics, then the firms default independently.


Symmetry considerations similar to those in the proof of Proposition~\ref{P:Model_I}
show that if \eqref{e6} holds for some choice of jump rates, then
there are constants $\lr_{m,n}$ and $\lu_{m,n}$, $0 \le m \le M$ and $0 \le n \le N$,
with $\lr_{M,n}=0$ for $0 \le n \le N$, $\lu_{m,N}=0$ for $0 \le m \le M$, and
$\lr_{m,n}$ and $\lu_{m,n}$ strictly positive otherwise such that
\begin{equation}
\label{ek_bi}
\begin{split}
(m+n)\alp'(t)+mn\be'(t)
& =r-\lr_{m,n}-\lu_{m,n} \\
& \quad + \frac{m}{M-m+1}\lr_{m-1,n}e^{-\alp(t)-n\be(t)} \\
& \quad + \frac{n}{N-n+1}\lu_{m,n-1}e^{-\alp(t)-m\be(t)},
\end{split}
\end{equation}
where we set $r:=\lu_{0,0}+\lr_{0,0}$
and adopt the convention that $\lr_{-1,n}=0$, $0 \le n \le N$,
and $\lu_{m,-1}=0$, $0 \le m \le M$.  We leave the straightforward details to the reader.

Setting $(m,n) = (1,0)$  in \eqref{ek_bi} gives
\begin{equation}
\label{E:alpha_first}
\alp'(t)=r-(\lr_{1,0}+\lu_{1,0})+\frac{\lr_{0,0}}{M}e^{-\alp(t)}.
\end{equation}
Similarly, setting $(m,n) = (0,1)$ in \eqref{ek_bi} gives
\begin{equation}
\label{E:alpha_second}
\alp'(t)=r-(\lr_{0,1}+\lu_{0,1})+\frac{\lu_{0,0}}{N}e^{-\alp(t)}.
\end{equation}
In particular, we have the identity
\begin{equation*}
\frac{\lr_{0,0}}{M}=\frac{\lu_{0,0}}{N}.
\end{equation*}

Setting $(m,n) = (1,1)$ in \eqref{ek_bi} and substituting in the expression for
$\alpha'(t)$ from \eqref{E:alpha_first} gives
\begin{equation}
\label{E:beta_model_II}
\begin{split}
\be' (t) &= r-(\lu_{1,1}+\lr_{1,1})+\left(\frac{\lr_{0,1}}{M}+\frac{\lu_{1,0}}{N}\right)e^{-\alp(t)-\be(t)}-2\alp'(t)\\
&  =r-(\lu_{1,1}+\lr_{1,1})+\left(\frac{\lr_{0,1}}{M}+\frac{\lu_{1,0}}{N}\right)e^{-\alp(t)-\be(t)}\\
& \quad - 2\left(r-(\lr_{1,0}+\lu_{1,0})+\frac{\lr_{0,0}}{M}e^{-\alp(t)}\right).\\
\end{split}
\end{equation}
Further substituting the above expressions for $\alp'(t)$ and $\be'(t)$ 
from \eqref{E:alpha_first} and \eqref{E:beta_model_II} into \eqref{ek_bi} 
for a general pair $(m,n)$ leads to a system of equations of the form
\begin{equation}
\label{cont_1}
\begin{split}
a_{m,n}e^{\alp(t)}+b_{m,n}& =c_{m,n}e^{-n\be(t)}+d_{m,n}e^{-m\be(t)} - e_{m,n}e^{-\be(t)}, \\
& \quad 0 \le m \le M, \, 0 \le n \le N, \\
\end{split}
\end{equation}
where the various coefficients are given by
\begin{eqnarray*}
\left\lbrace
 	\begin{aligned}
		a_{m,n}&=(m+n-2mn)(r-\lr_{1,0}-\lu_{1,0})+mn(r-\lu_{1,1}-\lr_{1,1}) \\
		       & \quad -r + \lu_{m,n}+\lr_{m,n}\\
		 b_{m,n}&=(m+n-2mn)\frac{\lr_{0,0}}{M}\\
        c_{m,n}&=\frac{m\lr_{m-1,n}}{M-m+1} \\
        d_{m,n}&=\frac{n\lu_{m,n-1}}{N-n+1} \\
        e_{m,n}&=mn\left(\frac{\lr_{0,1}}{M}+\frac{\lu_{1,0}}{N}\right).
	\end{aligned}
 \right.
\end{eqnarray*}
Note that $c_{m,n}> 0$ whenever $m>0$ and $d_{m,n}> 0$ whenever $n>0$.

We claim that the continuous function $\beta$ is constant. 
Assume without loss of generality that $M\geq 3$ and suppose that the function $\beta$ is not constant.

Consider the two cases:
\begin{itemize}
\item [(i)] One of $a_{2,1}$ or $a_{3,1}$ is zero.
\item [(ii)] Both $a_{2,1}$ and $a_{3,1}$ are non-zero.
\end{itemize}

Case (i) is impossible, because for either $m=2$ or $m=3$ we would have
$0 = d_{m,1} e^{-m \beta(t)} + (c_{m,1} - e_{m,1}) e^{-\beta(t)} - b_{m,1}$
and conclude that either a quadratic or cubic
polynomial was identically zero on some open interval.

Turning to Case (ii), note first that if $a_{m,n}\neq 0$ we have
\begin{equation}
\label{e_amn}
e^{\alp(t)}
=\frac{c_{m,n}e^{-n\be(t)}+d_{m,n}e^{-m\be(t)}-e_{m,n}e^{-\be(t)}-b_{m,n}}{a_{m,n}}
\end{equation}
and so we would have
a quadratic and a cubic polynomial 
that agreed on an  open interval.

Therefore, the function $\beta$ must be a constant, say $\beta^*$.  
Of course, $\beta^*$ is the pre-specified value for $\beta(T)$.
We now show that $\beta^* = 0$.

Equation \eqref{ek_bi} becomes
\begin{equation}
\label{ek_bi_constant}
\begin{split}
& (m+n)\left(r-(\lr_{1,0}+\lu_{1,0})+\frac{\lr_{0,0}}{M}e^{-\alp(t)}\right) \\
& =r-(\lr_{m,n}+\lu_{m,n})
+ \frac{m\lr_{m-1,n}}{M-m+1}e^{-\alp(t)-n\beta^*}+\frac{n\lu_{m,n-1}}{N-n+1}e^{-\alp(t)-m\beta^*}.\\
\end{split}
\end{equation}
For a fixed pair $(m, n)$, each side of \eqref{ek_bi_constant} is a first degree polynomial
in $\alp(t)$ and, since $\alp(t)$ is continuous and non-constant, we can equate
coefficients. If we also  record our boundary conditions and conventions from above
we arrive at the following system of equations for $0\leq m\leq M$ and $0\leq n \leq N$
\begin{eqnarray}
\label{R1}
\left\lbrace
 	\begin{aligned}
\frac{(m+n)}{M}\lr_{0,0}&=\frac{m}{M-m+1}\lr_{m-1,n}e^{-n\beta^*} \\
& \quad +\frac{n}{N-n+1}\lu_{m,n-1}e^{-m\beta^*}\\
		 (m+n)(r- \lr_{1,0}-\lu_{1,0})&= r- \lr_{m,n}-\lu_{m,n}\\
        \lu_{m,N}&=0 \\
        \lr_{M,n} &=0 \\
        \lu_{m,-1}&=0 \\
        \lr_{-1,n}&=0.
	\end{aligned}
 \right.
\end{eqnarray}

Because $\lr_{M,N}=\lu_{M,N}=0$, we see from the second equation of \eqref{R1} for
$(m,n) = (M,N)$ that
\begin{equation}
\label{E:1001}
r - \lr_{1,0} - \lu_{1,0} = \frac{r}{M+N},
\end{equation}
and we can substitute this value 
into the second equation of \eqref{R1} for general $(m,n)$ to conclude that
\[
\lr_{m,n}+\lu_{m,n}
 = 
r\left(1 - \frac{m+n}{M+N}\right), 
\quad 0\leq m\leq M, \, 0\leq n \leq N.
\]

Setting $m=0$ in the first equation
of \eqref{R1} and noting the fifth equation, we get
\begin{equation*}
\lu_{0,n-1} = \frac{N-n+1}{M} \lr_{0,0}, \quad 1\leq n\leq N.
\end{equation*}
A similar argument leads to
\begin{equation*}
\lr_{m-1,0} = \frac{M-m+1}{M} \lr_{0,0}, \quad 1\leq m\leq M.
\end{equation*}

Combining these observations, we arrive at the following system of equations
\begin{equation}\label{e_1.5r}
\left\lbrace
 	\begin{aligned}
        & \frac{n\lu_{m,n-1}}{N-n+1}e^{(n-m)\beta^*} - \frac{m\lu_{m-1,n}}{M-m+1} \\ 
        & \quad = \frac{\lr_{0,0}}{M}
        \Bigl[(m+n)e^{n\beta^*} 
            - \frac{m(M+N-m-n+1)}{M-m+1}\Bigr], \\
        & \qquad 1\leq m\leq M, \, 0\leq n \leq N,\\
		& \lr_{m,n}+\lu_{m,n} = r\left(1-\frac{m+n}{M+N}\right), \quad 0\leq m\leq M, \, 0\leq n \leq N,\\
    & \lu_{m,-1}=0, \quad 0\leq m \leq M,\\
    & \lr_{-1,n}=0, \quad 0\leq n \leq N,\\
    & \lr_{M,n}=0, \quad 0 \le n \le N,\\
    & \lu_{m,N}=0, \quad 0 \le m \le M,\\ 
    & \lu_{0,n} = \frac{N-n}{M} \lr_{0,0}, \quad 0\leq n\leq N, \\
    & \lr_{m,0} = \frac{M-m}{M} \lr_{0,0}, \quad 0\leq m\leq M.
 \end{aligned}
 \right.
\end{equation}

Note that if $\lr_{m,n}$ and $\lu_{m,n}$ satisfy \eqref{R1}, then they also satisfy \eqref{e_1.5r}.
It will thus suffice to show that if  $\beta^*\neq 0$, then there do not exist
$\lr_{m,n}$ and $\lu_{m,n}$ satisfying \eqref{e_1.5r}.

Setting $(m,n) = (1,1)$ in the first equation of \eqref{e_1.5r} gives
\begin{equation*}
\frac{\lu_{1,0}}{N} - \frac{\lu_{0,1}}{M}
=\frac{\lr_{0,0}}{M}\left(2e^{\beta^*}-\frac{M+N-1}{M}\right),
\end{equation*}
while the fifth equation of \eqref{e_1.5r} forces $\lu_{0,1}=\frac{\lr_{0,0}}{M}(N-1)$.
Thus,
\begin{equation}\label{C_1}
\lu_{1,0}=\frac{N}{M}\lr_{0,0}(2e^{\beta^*}-1).
\end{equation}

If instead we set $(m,n) = (2,0)$ in the first equation from \eqref{e_1.5r},  we obtain
\begin{equation}
\label{C_2}
\lu_{1,0}=\frac{N}{M}\lr_{0,0}.
\end{equation}

Comparing \eqref{C_1} and \eqref{C_2}, we conclude that
\begin{equation*}
2e^{\beta^*}-1=1,
\end{equation*}
and hence $\beta^*=0$.
\end{proof}

\section{Two classes with different individual propensity to default}
\label{S:Model_III}

Recall Model III from the Introduction.  As with Model II, the
graph $G$ is $K_{M,N}$, the complete bipartite 
graph with vertex set the disjoint union $V = \hat V \sqcup \check V$, where $\hat V$ has $M$ vertices and 
$\check V$ has $N$
vertices.  Now, however, there are functions $\hat \alpha$, $\check \alpha$ and $\beta$ such that
\begin{equation} \label{e12_III}
\left\lbrace
 	\begin{aligned}
		\alpha_u (t) &= \hat \alpha (t) ~\text{for all}~ u \in \hat V\\
		\alpha_v (t) &= \check \alpha (t) ~\text{for all}~ v \in \check V\\
		\beta_{uv} (t) &= \beta (t)~\text{for all}~ u \in \hat V,  \, v \in \check V.
	\end{aligned}
 \right.
 \end{equation}

\begin{proposition}
\label{p_a1a2}
Consider Model III with $M \geq 4, N\geq 3$ or $M\geq 3, N \geq 4$.
Suppose that the prescribed values of $\hat \alpha(T)$
and $\check \alpha(T)$ are distinct.  If the prescribed
value of $\beta(T)$ is non-zero and sufficiently small, then
the model does not admit time-homogeneous Markovian dynamics.
\end{proposition}

\begin{proof}
Another symmetry argument similar to those in the proofs of Proposition~\ref{P:Model_I}
and Proposition~\ref{p_aa} shows that if \eqref{e6} holds for some choice of jump rates, then
there are constants $\lr_{m,n}$ and $\lu_{m,n}$, $0 \le m \le M$ and $0 \le n \le N$,
with $\lr_{M,n}=0$ for $0 \le n \le N$, $\lu_{m,N}=0$ for $0 \le m \le M$, and
$\lr_{m,n}$ and $\lu_{m,n}$ strictly positive otherwise such that
\begin{eqnarray}\label{ek_bi_2}
m\hat \alp'(t)&+&n\check \alp'(t)+mn\be'(t)=r -(\lr_{m,n}+\lu_{m,n}) +\nonumber\\
&+&\frac{m\lr_{m-1,n}}{M-m+1}e^{-\hat \alp(t)-n\be(t)}+\frac{n\lu_{m,n-1}}{N-n+1}e^{-\check \alp(t)-m\be(t)},
\end{eqnarray}
where we set $r:=\lu_{0,0}+\lr_{0,0}$
and adopt the convention that $\lr_{-1,n}=0$, $0 \le n \le N$,
and $\lu_{m,-1}=0$, $0 \le m \le M$.

Applying \eqref{ek_bi_2} with $(m,n)=(1,0)$ and $(m,n)=(0,1)$ gives
\begin{equation}\label{e_alp}
\left\lbrace
 	\begin{aligned}
		\hat \alp'(t)&=r-(\lu_{1,0}+\lr_{1,0})+\frac{\lr_{0,0}}{M}e^{-\hat \alp(t)}\\
		 \check \alp'(t)&=r-(\lu_{0,1}+\lr_{0,1})+\frac{\lu_{0,0}}{N}e^{-\check \alp(t)}.
	\end{aligned}
 \right.
\end{equation}
Similarly, applying  \eqref{ek_bi_2} with $(m,n) = (1,1)$ and then substituting
in the expressions for $\hat \alp(t)$ and $\check \alp(t)$ from \eqref{e_alp} gives
\begin{equation}
\label{e_beta}
\begin{split}
\be' (t) &= r-(\lu_{1,1}+\lr_{1,1})+\frac{\lr_{0,1}}{M}e^{-\hat \alp(t)-\be(t)}+\frac{\lu_{1,0}}{N}e^{-\check \alp(t)-\be(t)}-\hat \alp'(t)-\check \alp'(t)\\
&= -r-(\lu_{1,1}+\lr_{1,1})-(\lu_{0,1}+\lr_{0,1})-(\lu_{1,0}+\lr_{1,0})+\frac{\lr_{0,1}}{M}e^{-\hat \alp(t)-\be(t)}\\
& \quad +\frac{\lu_{1,0}}{N}e^{-\check \alp(t)-\be(t)}-\frac{\lr_{0,0}}{M}e^{-\hat \alp(t)} -\frac{\lu_{0,0}}{N}e^{-\check \alp(t)}. \\
\end{split}
\end{equation}

Substituting the expressions for the $\hat \alp$, $\check \alp$ and $\be$ from
\eqref{e_alp} and \eqref{e_beta} into \eqref{ek_bi_2} for general $(m,n)$
produces a system of equations of the form
\begin{equation}
\label{e_a1a2}
\begin{split}
& a_{m,n}e^{\hat \alp(t)}
+ b_{m,n} 
+ c_{m,n}e^{\hat \alp(t)-\check \alp(t)}
+ d_{m,n}e^{-\be(t)}
+ e_{m,n}e^{\hat \alp(t)-\check \alp(t)-\be(t)} \\
& \quad = f_{m,n}e^{-n\be(t)} +g_{m,n}e^{\hat \alp(t)-\check \alp(t)-m\be(t)}, \\
\end{split}
\end{equation}
where
\begin{equation*}
\left\lbrace
 	\begin{aligned}
		a_{m,n}&=(m-mn)(r-\lu_{1,0}-\lr_{1,0})+(n-mn)(r-\lu_{0,1}-\lr_{0,1}) \\
		       & \quad +mn(r-\lu_{1,1}-\lr_{1,1}) -r +\lr_{m,n}+\lu_{m,n}\\
		b_{m,n}&=(m-mn)\frac{\lr_{0,0}}{M}\\
        c_{m,n}&=(n-mn)\frac{\lu_{0,0}}{N}\\
        d_{m,n}&=mn\frac{\lr_{0,1}}{M}\\
        e_{m,n}&=mn\frac{\lu_{1,0}}{N}\\
        f_{m,n}&=\frac{m\lr_{m-1,n}}{M-m+1}\\
        g_{m,n}&=\frac{n\lu_{m,n-1}}{N-n+1}.
	\end{aligned}
 \right.
\end{equation*}
Observe that because 
$\lr_{m,n}$  is strictly positive for $1 \leq m\leq M-1$ and $1\leq n \leq N$
and 
$\lu_{m,n}$  is strictly positive for $1 \leq m\leq M$ and $1\leq n \leq N-1$,
both
$f_{m,n}$ and $g_{m,n}$ are strictly positive for $1 \leq m\leq M$ and $1\leq n \leq N$.

We claim that the continuous function $\beta$ is constant.
Assume without loss of generality that $M \ge 4$, $N \ge 3$ and suppose that the function
$\be$ is not constant. 

Re-arrange \eqref{e_a1a2} to get
\begin{equation}
\label{e_x}
\begin{split}
& \left[
a_{m,n}+c_{m,n}e^{-\check \alp(t)}+e_{m,n}e^{-\check \alp(t)}e^{-\be(t)}-g_{m,n}e^{-\check \alp(t)} e^{-m\be(t)}
\right]
e^{\hat \alp(t)} \\
& \quad =
f_{m,n} e^{-n \be(t)} -d_{m,n}e^{-\be(t)}-b_{m,n}. \\
\end{split}
\end{equation}
Because $f_{m,n}$ is strictly positive for 
$1 \leq m\leq M$ and $1\leq n \leq N$, there is an open interval
$J$ such that the image $\{\be(t) : t \in J\}$
contains an open interval and the right-hand side of \eqref{e_x}
is non-zero for all $t \in J$, $1 \leq m \leq M$ and $2 \leq n \leq N$,
and hence the same is true for the left-hand side.

Taking \eqref{e_x} with the indices $(m,n)$ replaced by another pair $(i,j)$, we see that if
$t \in J$, $1 \leq m,i \leq M$ and $2 \leq n,j \leq N$, then
\begin{equation}
\label{e_xx}
\begin{split}
&\left[f_{m,n}e^{-n \be(t)}-d_{m,n}e^{-\be(t)}-b_{m,n}\right] \\
& \quad \times 
\left[a_{i,j}+c_{i,j}e^{-\check \alp(t)}+e_{i,j}e^{-\check \alp(t)}e^{-\be(t)}-g_{i,j}e^{-\check \alp(t)} e^{-i \be(t)}\right] \\
& \qquad = 
\left[f_{i,j} e^{-j \be(t)}-d_{i,j}e^{-\be(t)}-b_{i,j}\right] \\
& \qquad \quad \times
\left[a_{m,n}+c_{m,n}e^{-\check \alp(t)}+e_{m,n}e^{-\check \alp(t)}e^{-\be(t)}-g_{m,n}e^{-\check \alp(t)}e^{-m\be(t)}\right]. \\
\end{split}
\end{equation}

Re-arranging \eqref{e_xx} gives 
\begin{equation}
\label{e_y}
p(e^{-\be(t)}; m,n,i,j) e^{-\check \alp(t)} = q(e^{-\be(t)}; m,n,i,j),
\end{equation}
where
\[
\begin{split}
p(z;m,n,i,j) & := (c_{i,j}+e_{i,j}z-g_{i,j}z^i)(f_{m,n}z^n-d_{m,n}z-b_{m,n}) \\ 
& \quad -(c_{m,n}+e_{m,n}z-g_{m,n}z^m)(f_{i,j}z^j-d_{i,j}z-b_{i,j}) \\
\end{split}
\]
and
\[
q(z;m,n,i,j) := a_{m,n}(f_{i,j}z^j-d_{i,j}z-b_{i,j})-a_{i,j}(f_{m,n}z^n-d_{m,n}z-b_{m,n}).
\]
Suppose now that $2 \leq m,i \leq M$ and $2 \leq n,j \leq N$.
The leading term of the polynomial $p(z;m,n,i,j)$ is $-g_{i,j} f_{m,n} z^{i+n}$ if $i+n > m+j$
and $g_{m,n} f_{i,j} z^{m+j}$ if $i+n < m+j$ (recall that $f_{m,n}, g_{m,n}, f_{i,j}, g_{i,j}$
are all strictly positive).  Therefore, by taking a subinterval of $J$ if necessary, 
when $i+n \ne m+j$ we may suppose that
$J$ retains the properties required of it above and, moreover, that
both sides of \eqref{e_y} are non-zero for all $t \in J$.
In particular, either $a_{m,n} \ne 0$ or $a_{i,j} \ne 0$ and the polynomial
$q(z;m,n,i,j)$ has degree either $n$ or $j$ when $n \ne j$.

Consider two 4-tuples $(m',n',i',j')$ and $(m'',n'',i'',j'')$ with
\begin{equation}\label{e_ind}
\left\lbrace
 	\begin{aligned}
 	& 2 \le m', m'', i', i'' \le M \\
 	& 2 \le n', n'', j', j'' \le N \\
 	& i'+n' \ne m'+j'\\
 	& i''+n'' \ne m''+ j'' \\
 	& n' \ne j' \\
 	& n'' \ne j''.
	\end{aligned}
 \right.
\end{equation}	
We conclude from \eqref{e_y} that
\begin{equation}
\label{E:4_index}
p(z; m',n',i',j') q(z; m'',n'',i'',j'') = p(z; m'',n'',i'',j'') q(z; m',n',i',j')
\end{equation}
for all $z$ in an open interval.  The left-hand side of \eqref{E:4_index} is a polynomial
in $z$ of degree either 
$((i'+n') \vee (m'+j')) + n''$ or $((i'+n') \vee (m'+j')) + j''$, 
whereas the right-hand size has degree either 
$((i''+n'') \vee (m''+j'')) + n'$ or $((i''+n'') \vee (m''+j'')) + j'$.
For $(m',n',i',j')=(2,2,2,3)$ and $(m'',n'',i'',j'')=(4,2,4,3)$ we have
\[
\left\lbrace
 	\begin{aligned}
 	&i'+n' = 4 \ne 5 = m'+j' \\
 	&i''+n'' = 6 \ne 7 = m''+j'' \\
 	&n' = 2 \ne 3 = j' \\
 	&n'' = 2 \ne 3 = j'' \\
 	&((i'+n') \vee (m'+j')) + n'' = 5 + 2 = 7\\
 	&((i'+n') \vee (m'+j')) + j'' = 5 + 3 = 8 \\
 	&((i''+n'') \vee (m''+j'')) + n' = 7 + 2 = 9 \\
 	&((i''+n'') \vee (m''+j'')) + j' = 7 + 3 = 10,
	\end{aligned}
 \right.
\]
and so the possible degrees of the left-hand side of \eqref{E:4_index} are $7$ and $8$,
whereas the possible degrees of the right-hand side are $9$ and $10$.

Therefore, the function $\beta$ must be a constant, say $\beta^*$.  
Note that $\beta^*$ is just the pre-specified value for $\beta(T)$.
We now show that $\beta^*=0$.

For the moment, consider Model III with $\beta(T) = 0$, so that the function
$\beta$ must be identically zero and the firms evolve independently.  
In this special case, we know from the Introduction that 
\begin{equation}
\label{E:hatalpha0}
\exp(\hat \alpha(t)) = (1 + \exp(\hat \alpha(T)))^{\frac{t}{T}} - 1
\end{equation}
and
\begin{equation}
\exp(\check \alpha(t)) = (1 + \exp(\check \alpha(T)))^{\frac{t}{T}} - 1.
\end{equation}
It follows from the linear independence of the functions 
$\exp(c_1 \cdot), \ldots, \exp(c_h \cdot)$
when $c_1, \ldots, c_h$ are distinct that in this case the functions
$\exp(\hat \alpha)$, $\exp(\check \alpha)$ and $\exp(\hat \alpha + \check \alpha)$
are linearly independent when $\hat \alpha(T) \ne \check \alpha(T)$.

Now return to the case of a general value for $\beta^* = \beta(T)$.
Equation \eqref{ek_bi_2} becomes
\begin{equation}
\label{ek_bi_2_const}
\begin{split}
m\hat \alp'(t)+n\check \alp'(t)
& =r-(\lr_{m,n}+\lu_{m,n}) \\
& \quad +\frac{m\lr_{m-1,n}}{M-m+1}e^{-\hat \alp(t)-n\beta^*}
+\frac{n\lu_{m,n-1}}{N-n+1}e^{-\check \alp(t)-m\beta^*}. \\
\end{split}
\end{equation}
We can solve \eqref{e_alp} for the $\hat \alpha$ and $\check \alpha$ as in
Section~\ref{S:generalities} to get 
\begin{equation}\label{e_alpe}
\left\lbrace
 	\begin{aligned}
\hat \alp(t)
&=\log\left(\frac{\lr_{0,0}}{M(r-\lu_{1,0}-\lr_{1,0})}(e^{(r-\lu_{1,0}-\lr_{1,0})t}-1)\right)\\
\check \alp(t)
		 &=\log\left(\frac{\lu_{0,0}}{N(r-\lu_{0,1}-\lr_{0,1})}(e^{(r-\lu_{0,1}-\lr_{0,1})t}-1)\right)
	\end{aligned}
 \right.
\end{equation}
Substituting \eqref{e_alpe} into \eqref{ek_bi_2_const} gives
\begin{equation}
\label{E:alph_in}
\begin{split}
&m(r-\lu_{1,0}-\lr_{1,0})+n(r-\lu_{0,1}-\lr_{0,1})+(\lr_{m,n}+\lu_{m,n})-r \\
& \quad + \left(\frac{m\lr_{0,0}}{M}-\frac{m\lr_{m-1,n}e^{-n\beta^*}}{M-m+1}\right)\frac{M(r-\lu_{1,0}-\lr_{1,0})}{\lr_{0,0}(e^{(r-\lu_{1,0}-\lr_{1,0})t}-1)}\\ 
& \quad + \left(\frac{n\lu_{0,0}}{N}-\frac{n\lu_{m,n-1}e^{-m\beta^*}}{N-n+1}\right)\frac{N(r-\lu_{0,1}-\lr_{0,1})}{\lu_{0,0}(e^{(r-\lu_{0,1}-\lr_{0,1})t}-1)}
=0. \\
\end{split}
\end{equation}

It follows from the observations above and a compactness argument that if 
$\hat \alpha(T) \ne \check \alpha(T)$ and
$\beta^* = \beta(T)$ is sufficiently close to zero, then the functions
$\exp(\hat \alpha)$, $\exp(\check \alpha)$ and $\exp(\hat \alpha + \check \alpha)$
are linearly independent.  Suppose that this is the case.
Equation \eqref{E:alph_in} is of the form
\begin{equation*}
a + b e^{-\hat \alpha(t)} + c e^{-\check \alpha(t)} = 0
\end{equation*}
for suitable constants $a,b,c$,
and hence $a=b=c=0$.
Thus,  
\begin{eqnarray*}
\frac{m\lr_{0,0}}{M} - \frac{m\lr_{m-1,n}e^{-n\beta^*}}{M-m+1}&=&0\\
\frac{n\lu_{0,0}}{N} - \frac{n\lu_{m,n-1}e^{-m\beta^*}}{N-n+1}&=&0\\
m(r-\lr_{1,0}-\lu_{1,0})+n(r-\lr_{0,1}-\lu_{0,1})+(\lr_{m,n}+\lu_{m,n})-r &=&0
\end{eqnarray*}
for $(m,n)\in \{0,\dots,M\}\times\{0,\dots,N\}$,
and so, after some algebra,
\[
\lr_{0,0}\frac{M-m}{M}e^{n\beta^*}
+
\lu_{0,0}\frac{N-n}{N}e^{m\beta^*}
= 
r-m(r-\lr_{1,0}-\lu_{1,0})-n(r-\lr_{0,1}-\lu_{0,1})
\]
for $(m,n)\in \{0,\dots,M\}\times\{0,\dots,N\}$.
In particular, considering $(m,n) = (k,k)$ for $0 \le k \le M \wedge N$
leads to a system of the form
\[
A k e^{k \beta^*} + B e^{k \beta^*} + C k + D = 0
\]
for suitable constants $A,B,C,D$ with $A>0$ and $B>0$.
A straight line can intersect the graph of the function
$t \mapsto A t e^{\beta^* t} + B e^{\beta^* t}$ at most
twice if $\beta^* > 0$ and at most three times if $\beta^* < 0$,
and since $M \wedge N \ge 3$ we must have $\beta^*=0$.
\end{proof}


\providecommand{\bysame}{\leavevmode\hbox to3em{\hrulefill}\thinspace}
\providecommand{\MR}{\relax\ifhmode\unskip\space\fi MR }
\providecommand{\MRhref}[2]{%
  \href{http://www.ams.org/mathscinet-getitem?mr=#1}{#2}
}
\providecommand{\href}[2]{#2}

\end{document}